\documentclass[a4, 11pt]{article}
\usepackage{wakaiki_singlecolumn_Eng}

\newcommand{\Int}{\mathop{\rm Int~\!}}
\newcommand{\Cl}{\mathop{\rm Cl~\!}}
\usepackage{graphicx, subcaption}      

\title{
Quantized Feedback Stabilization of Sampled-Data Switched Linear Systems
}
\author{Masashi Wakaiki and Yutaka Yamamoto
\thanks{M. Wakaiki and Y. Yamamoto are with the Department of Applied Analysis and Complex
Dynamical Systems, Graduate School of Informatics, Kyoto University, Kyoto
606-8501, Japan
(e-mail:{\tt  \ wakaiki@acs.i.kyoto-u.ac.jp};
{\tt \ yy@i.kyoto-u.ac.jp}).}%
}
\begin{document}
\maketitle
\thispagestyle{empty}
\pagestyle{empty}

\begin{abstract}                
We propose a stability analysis method for sampled-data switched linear
systems with quantization.
The available information to the controller
is limited: the quantized state and switching signal at
each sampling time.
Switching between sampling times can produce
the mismatch of the modes between the plant and the controller.
Moreover, the coarseness of quantization makes the trajectory wander
around, not approach, the origin.
Hence the trajectory may leave the desired neighborhood
if the mismatch leads to instability of the closed-loop system.
For the stability of the switched systems,
we develop a sufficient condition characterized by the {\em total mismatch time}.
The relationship between the mismatch time and the dwell time 
of the switching signal is also discussed.
\end{abstract}


\section{Introduction}
In this paper, we consider a
sampled-data switched linear system with a memoryless quantizer 
in Fig.~\ref{fig:WSDSLS}.
The available information to the controller is only
the quantized state and switching signal at each sampling time.
We then raise the questions:
{\em
What conditions are needed for the stability of the closed-loop system
under such imcomplete information?
If the system is stable, how close can the trajectories get to the origin?}

Switched systems and quantized control have been studied extensively but separately;
see, e.g., \cite{Liberzon2003Book, Lin2009} for switched systems
and \cite{Ishii2002Book, Nair2007} for quantized control.
Few works examine 
the state behavior of a switched system with quantization and
the effect of switching between sampling times. 
Recently, \cite{Liberzon2014} has proposed an
encoding and control strategy that achieves
sampled-data quantized state feedback stabilization of switched systems.
This strategy is rooted in
the non-switched case in \cite{Liberzon2003}.
In \cite{Liberzon2014}, the input of the controller is a
discrete-valued and discrete-time signal, whereas
the controller generates 
a continuous-valued and continuous-time output signal.
In contrast, here we consider a controller whose {\em output as well as input} are
{\em discrete-valued} and {\em discrete-time} signals.

 \begin{figure}[b]
 \centering
 \includegraphics[width = 7cm,bb= 60 175 700 460,clip]{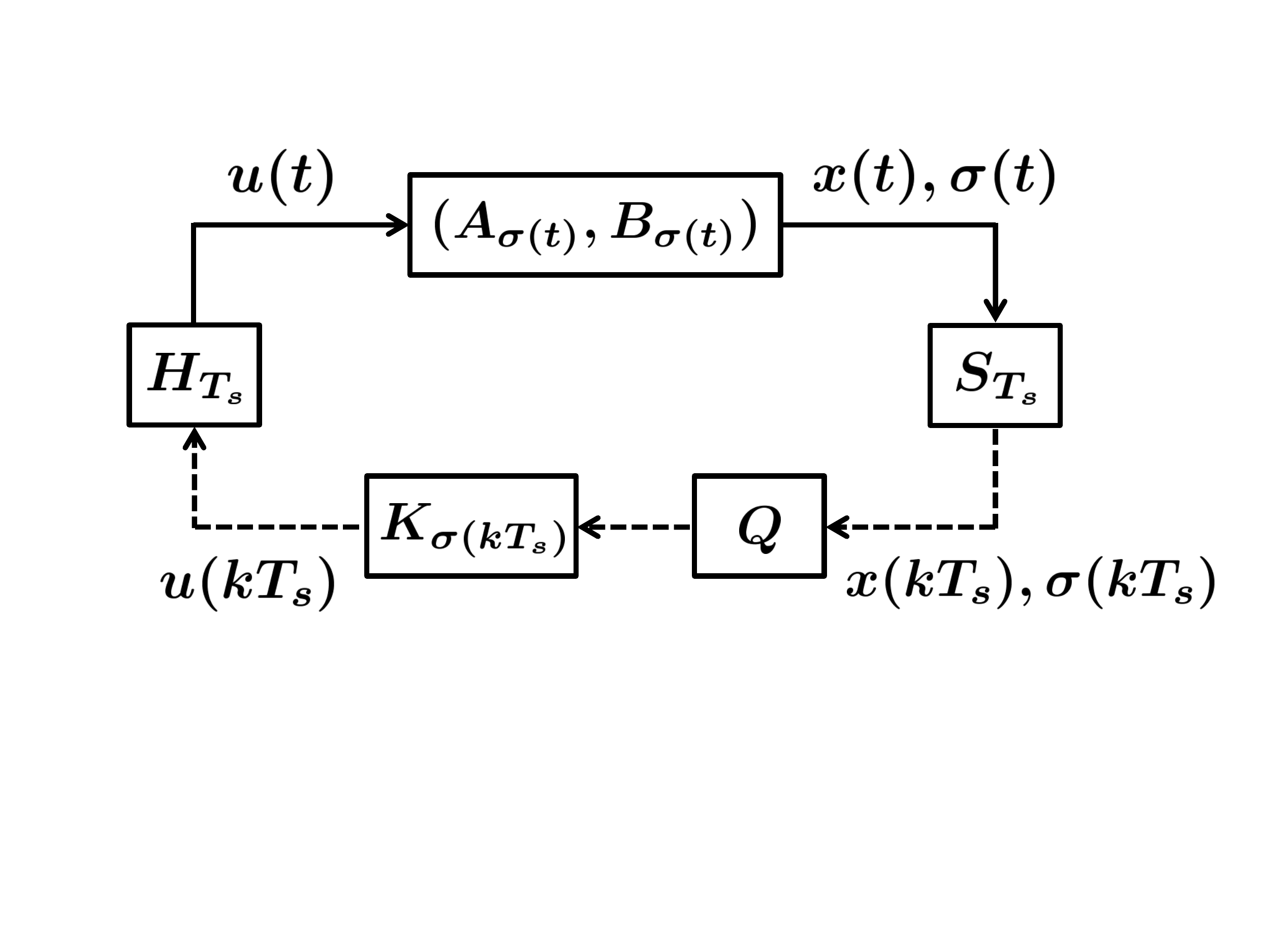}
 \caption{Sampled-data switched system with quantization, where $T_s$ is a sampling period}
 \label{fig:WSDSLS}
 \end{figure}

\cite{Ishii2002Book,Ishii2004} have studied the stability analysis
of a sampled-data non-switched system of a memoryless quantizer.
Since a memoryless quantizer does not give an accurate value of the state near the origin, 
asymptotic stability cannot be generally achieved. 
However, 
such a quantizer is useful because of the simplicity in implementation.
\cite{Ishii2004} have developed a sufficient condition
for a non-switched system to be quadratic attractive. The authors have also
provided a randomized algorithm to verify this stability property in a computationally efficient way.
In the present paper, we use this algorithm and a scheduling function 
with the revisitation property introduced in \cite{Liberzon2004}.
The combined method constructs  a common Lyapunov function guaranteeing
the quadratic attractiveness of each subsystem.
%
%

We face two challenges in the stability analysis of the switched system in Fig.~\ref{fig:WSDSLS}.
First, since only at sampling times we know which subsystem is active, 
we do not always use the feedback gain designed for the subsystem active at
the present time.
Therefore
the closed-loop system may become unstable when switching occurs between sampling times.
Second, after arriving at a certain neighborhood of the origin,
the trajectory may not approach the origin anymore due to the coarseness of quantization. 
This implies that the trajectory can leave
the desired neighborhood if switching makes the system unstable.

This paper is organized as follows. 
In Section~2, we state the switched system and 
the information structure together with basic assumptions.
In Section~3, we first investigate the growth rate of the common
Lyapunov function when switching occurs in a sampling interval.
Next we develop
a stability analysis method for the sampled-data switched system 
by using
the {\em total mismatch time}, the total time when the modes 
mismatch between the plant and the controller.
In Section~4, we briefly discuss the relationship between the mismatch time
and the dwell time of the switching signal.
Section~5 concludes this paper.

\noindent
{\bf Notation} \\
We denote by $\mathbb{Z}_+$ the set of non-negative integers
$\{k \in \mathbb{Z}:~k \geq 0\}$.
For a set $\Omega \subset \mathbb{R}^{\sf n}$, $\Cl (\Omega)$, 
$\Int (\Omega)$, and $\partial \Omega$ are its closure,
interior, and boundary, respectively.

Let $M^{\top}$ denote the transpose of $M \in \mathbb{R}^{\sf n\times m}$.
The Euclidean norm of $v \in \mathbb{R}^{\sf n}$ is defined by
$\|v\| = (v^{\top}v)^{1/2}$.
For $M \in \mathbb{R}^{\sf m\times n}$, its Euclidean induced norm is
defined by $\|M\| = \sup \{ \|Mv\|:~v\in \mathbb{R}^{\sf n},~\|v\|= 1 \}$ and
equals the largest singular value of $M$.
Let $\lambda_{\max}(P)$ and $\lambda_{\min}(P)$ denote
the largest and the smallest eigenvalue of $P \in \mathbb{R}^{\sf n\times n}$.

Let $T_s$ be a sampling period. For $t \geq 0$,
we define $[t]^-$ by
\begin{equation*}
[t]^- = kT_s\qquad \text{if~~~}kT_s \leq t < (k+1)T_s \quad (k \in \mathbb{Z}_+).
\end{equation*}

\section{Sampled-data Switched Systems with Quantization}
\subsection{Switched systems}
Consider the continuous-time switched linear system
\begin{equation}
\label{eq:SLS}
\dot x = A_{\sigma}x + B_{\sigma}u, 
\end{equation}
where $x(t) \in \mathbb{R}^{\sf n}$ is the state and 
$u(t) \in \mathbb{R}^{\sf m}$ is the control input.
For a finite index set $\mathcal{P}$, the mapping 
$\sigma:~[0,\infty) \to \mathcal{P}$ is right-continuous and piecewise constant.
We call $\sigma$ \textit{switching signal} and the discontinuities of $\sigma$ 
\textit{switching times}.

We assume that all subsystems are stabilizable and that
only finitely many switches occur on any finite interval:
\begin{assumption}
\label{ass:system}
For every $p \in \mathcal{P}$, $(A_p, B_p)$ is stabilizable, i.e., 
there exists $K_p \in \mathbb{R}^{\sf m \times n}$ such that
$A_p+B_pK_p$ is Hurwitz. 
Furthermore, every sampling interval has 
at most one switch.
\end{assumption}

\subsection{Quantized sampled-data system}
Let $T_s > 0$ be a sampling period.
The sampler $S_{T_s}$ is given by
\begin{equation*}
S_{T_s}:~ (x,\sigma) \mapsto (x(kT_s),\sigma(kT_s))\qquad (k \in \mathbb{Z}_+)
\end{equation*}
and
the zero-th hold $H_{T_s}$ by
\begin{equation*}
H_{T_s}:~u_d \mapsto u(t)=u_d(k),~~ t\in[kT_s,(k+1)T_s)~~ (k \in \mathbb{Z}_+).
\end{equation*}

We now state the definition of a memoryless quantizer $Q$ given in \cite{Ishii2004}.
For an index set $\mathcal{S}$,
the partition $\{\mathcal{Q}_j\}_{j \in \mathcal{S}}$ 
of $\mathbb{R}^{\sf n}$ is 
said to be \textit{finite} if for every bounded set $B$, there exists a
finite subset $\mathcal{S}_f$ of $\mathcal{S}$ such that
$B \subset \bigcup_{j \in \mathcal{S}_f} \mathcal{Q}_j$.
We define
the quantizer $Q$ 
with respect to the finite partition $\{\mathcal{Q}_j\}_{j \in \mathcal{S}_f}$
by
\begin{align*}
Q:~\mathbb{R}^{\sf n} &\to 
\{q_j \}_{j \in \mathcal{S}_f} \subset \mathbb{R}^{\sf n} \\
x &\mapsto q_j \quad \text{if~~} x \in \mathcal{Q}_j\quad (j \in \mathcal{S}_f).
\end{align*}

The second assumption is that $Q(x) = 0$ if $x$ is close to the origin. 
\begin{assumption}
\label{ass:quantization_near_origin}
If $\Cl (\mathcal{Q}_j)$ contains the origin, then $q_j = 0$.
\end{assumption}

Let $q_x$ be the output of the zero-th hold whose input is 
the quantized state at sampling times, i.e.,
$
q_x(t) = Q(x([t]^-)).
$
Note that in Fig. \ref{fig:WSDSLS}, the control input $u$ is given by
\begin{equation}
\label{eq:control_input}
u(t) = K_{\sigma([t]^-)}q_x(t).
\end{equation}

Let $P \in \mathbb{R}^{\sf n \times n}$ be positive-definite and 
define the quadratic Lyapunov function
$V(x) = x^{\top} P x$ for $x \in \mathbb{R}^{\sf n}$.
Its time derivative $\dot V$ along the trajectory of \eqref{eq:SLS} with \eqref{eq:control_input}
is given by
\begin{align}
\dot V(&x(t),q_x(t),\sigma(t))
= (A_{\sigma(t)}x(t)+B_{\sigma(t)}K_{\sigma([t]^-)}q_x(t))^{\top}Px(t)
\label{eq:dotV_def}
+x(t)^{\top}P(A_{\sigma(t)}x(t)+B_{\sigma(t)}K_{\sigma([t]^-)}q_x(t))
\end{align}
if $t$ is not a switching time or a sampling time.

For $p,q \in \mathcal{P}$ with $p \not= q$, 
we also define $\dot V_p$ and $\dot V_{p,q}$ by
\begin{align}
\dot V_p(x(t),q_x(t))
&=
(A_px(t)+B_pK_pq_x(t))^{\top}Px(t) 
 +x(t)^{\top}P(A_px(t)+B_pK_pq_x(t)) \notag \\
\dot V_{p,q}(x(t),q_x(t)) 
&=
(A_px(t)+B_pK_qq_x(t))^{\top}Px(t) 
+x(t)^{\top}P(A_px(t)+B_pK_qq_x(t)). 
\label{eq:Vpq_def}
\end{align}
Then
$\dot V_{p}$ and $\dot V_{p,q}$ are the time derivatives of $V$
along the trajectories of the systems $(A_p, B_pK_p)$ and $(A_p, B_pK_q)$, 
respectively.

Every individual mode is assumed to be 
stable in the following sense with
a common Lyapunov function:
\begin{assumption}
\label{ass:subsystem_QA}
Consider 
the following sampled-data non-switched systems with quantization:
\begin{equation}
\label{eq:subsystem_QSDS}
\dot x = A_px + B_pu, \quad u = K_pq_x \qquad (p \in \mathcal{P}).
\end{equation}
Let $C$ be a positive number and
suppose that $R$ and $r$ satisfy $R > r > 0$. Then
there exists a positive-definite matrix $P \in \mathbb{R}^{\sf n \times n}$
such that for all $p \in \mathcal{P}$, every trajectory $x$ of the system
\eqref{eq:subsystem_QSDS} with $x(0) \in \overline{\mathcal{E}}_P(R)$
satisfies
\begin{equation}
\label{eq:dotVp_bound}
\dot V_p(x(t),q_x(t)) \leq -C \|x(t)\|^2
\end{equation}
or $x(t) \in \underline{\mathcal{E}}_P(r)$ for $t \geq 0$,
where $\overline{\mathcal{E}}_P(R)$ and $\underline{\mathcal{E}}_P(r)$
are given by
\begin{align*}
\overline{\mathcal{E}}_P(R) &= 
\{
x \in \mathbb{R}^{\sf n}:~V(x) \leq R^2\lambda_{\max}(P)
\}\\
\underline{\mathcal{E}}_P(r) &= 
\{
x \in \mathbb{R}^{\sf n}:~V(x) \leq r^2\lambda_{\min}(P)
\}.
\end{align*}
\end{assumption}

Assumption \ref{ass:subsystem_QA} implies the followings:
If we have no switches, then the
common Lyapunov function $V$ decreases
at a certain rate until
$V \leq r^2\lambda_{\min}(P)$. 
Furthermore,
$\underline{\mathcal{E}}_P(r)$
as well as
$\overline{\mathcal{E}}_P(R)$ are invariant sets.

The objective of the present paper is to
find switching conditions for the switched system in Fig.~\ref{fig:WSDSLS} 
to arrive at some neighborhood of the origin and remain there. 
We also determine how small the neighborhood is.

\begin{remark}
{\bf (a)}
Let $\mathcal{B}(L)$ be
the closed ball in $\mathbb{R}^{\sf n}$ with center at 0 and radius $L$. 
The ellipsoid $\overline{\mathcal{E}}_P(R)$ is the \textit{smallest} level set of $V$
\textit{containing} $\mathcal{B}(R)$, whereas $\underline{\mathcal{E}}_P(r)$ is
the \textit{largest} level set of $V$ \textit{contained in} $\mathcal{B}(r)$. 

\noindent
{\bf (b)}
In the non-sampled case,
the existence of common Lyapunov functions is a sufficient condition 
for stability under arbitrary switching; see, e.g., \cite{Liberzon2003Book, Lin2009}.
For sampled-data switched systems, however,
such functions do not guarantees the stability because
a switch within a sampling interval may make the closed-loop system unstable. 

\noindent
{\bf (c)}
For plants with a single mode,
\cite{Ishii2004}
proposed a randomized algorithm for the computation of $P$ in 
Assumption~\ref{ass:subsystem_QA}.
Combining the algorithm with a scheduling function that has the 
revisitation property in \cite{Liberzon2004},
we can efficiently compute the desired common Lyapunov function.
Since this is an immediate consequence of the above two works,
we omit the details.

\noindent
{\bf (d)}
Assumption \ref{ass:subsystem_QA} does not cover
the trajectory after switches even without mode mismatch.
If the mode changes $p_1 \to p_2 \to p_1$ 
at the switching times $t_1$ and $t_2$ 
on a sampling interval $(kT_s, (k+1)T_s)$,
then \eqref{eq:dotVp_bound} holds only for $t \in (T,t_1)$.
In Assumption \ref{ass:system}, we therefore assume that 
at most one switch occurs on
a sampling interval.
%
\end{remark}

\section{Stabilization with Limited Information}
\subsection{Upper bounds of $\dot V_{p,q}$}
Assumption~\ref{ass:subsystem_QA} gives an upper bound \eqref{eq:dotVp_bound}
of $\dot V_{p}$, i.e., $\dot V$ when we use the feedback gain designed for
the currently active subsystem. 
In this subsection,
we will
find an upper bound of $\dot V_{p,q}$, i.e.,
$\dot V$ when intersample switching leads to the
mismatch of the modes between the plant and the feedback gain.
To this end, we investigate the state behavior in sampling intervals.

Let us first examine the relationship among 
the original state $x(t)$, the sampled state $x([t]^-)$,
and the sampled quantized state $q_x(t)$.

The partition $\{\mathcal{Q}_j\}_{j \in \mathcal{S}_f}$ is finite. 
Moreover, 
Assumption \ref{ass:quantization_near_origin} shows that
if $\xi_{k} \to 0$ ($k \to \infty$) for some sequence $\{\xi_k\} \subset \mathcal{Q}_j$,
then $Q(x) = 0$ for all $x \in \mathcal{Q}_j$.
Hence
there exists $\alpha_0 > 0$ such that
\begin{align}
\|B_pK_q Q(x)\| \leq \alpha_0 \|x\|, \label{eq:alpha0_bound}
\end{align}
for $p,q \in \mathcal{P}$ and $x \in
\overline{\mathcal{E}}_P(R)$.
We also define $\Lambda$ by 
\[
\Lambda = \max_{p \in \mathcal{P}} \|A_p\|.
\]

The next result gives an upper bound on
the norm of the sampled state $x([t]^-)$ with
the original state $x(t)$.

\begin{lemma}
\label{lem:alpha1_bound}
Consider the swithced system \eqref{eq:SLS} with
\eqref{eq:control_input}, where $\sigma$ has finitely many switching times
on every finite interval.
Suppose that
\begin{equation}
\label{eq:alpha_0_condition}
\eta := \alpha_0 \frac{e^{\Lambda T_s} - 1}{\Lambda} < 1,
\end{equation}
and define $\alpha_1$ by
\begin{align*}
\alpha_1 = 
\frac{e^{\Lambda T_s}}{1 - \eta}.
\end{align*}
Then 
we have
\begin{align}
\label{eq:alpha1_bound}
\| x([t]^-)\| <
\alpha_1 \|x(t)\|
\end{align}
for all $t \geq 0$ with 
$x([t]^-) \in
\overline{\mathcal{E}}_P(R)$.
\end{lemma}
\begin{proof}
It suffices to prove \eqref{eq:alpha1_bound} for
$x(0) \in \overline{\mathcal{E}}_P(R)$ and 
$t \in [0,T_s)$.

Let $\Phi(\tau_1,\tau_2)$ denote the state-transition matrix of 
the switched system \eqref{eq:SLS}
for $\tau_1 \geq \tau_2$.
If a switch does not occur, $\Phi(\tau_1,\tau_2)$ is given by 
$\Phi(\tau_1,\tau_2) = e^{(\tau_1-\tau_2)A_{\sigma(0)}}$.
If $t_1,t_2,\dots, t_m$ are switching times on an interval $[\tau_2,\tau_1)$
and if we define $t_0=\tau_2$ and $t_{m+1}=\tau_1$, then
we have
\begin{align*}
\Phi(\tau_1,\tau_2) = 
\prod_{k=0}^{m} e^{(t_{k+1} - t_k)A_{\sigma(t_k)} }
\end{align*}
Since
\begin{equation}
\label{eq:state_t}
x(t) = \Phi(t,0) x(0) + \int^{t}_0 \Phi(t,\tau)
B_{\sigma(\tau)} K_{\sigma(0)} q_x(\tau) d \tau
\end{equation}
and since $\Phi(\tau,0)^{-1} = \Phi(t,0)^{-1}\Phi(t,\tau)$, it follows that
\begin{equation*}
x(0) = \Phi(t,0)^{-1} x(t) + \int^{t}_0 \Phi(\tau,0)^{-1}
B_{\sigma(\tau)} K_{\sigma(0)} q_x(\tau) d \tau.
\end{equation*}
This leads to
\begin{align}
\|x(0)\| \leq &\|\Phi(t,0)^{-1}\| \cdot \|x(t)\| + \left\| \int^{t}_0 \Phi(\tau,0)^{-1}
B_{\sigma(\tau)} K_{\sigma(0)} q_x(\tau) d \tau \right\|.
\label{eq:x0_bound}
\end{align}
Let $t_1,t_2,\dots, t_m$ be switching times on the interval $[0,t)$.
Since $\| e^{\tau A}\| \leq e^{\tau\|A\|}$ for $\tau \geq 0$,
we obtain
\begin{align}
&\|\Phi(t,0)^{-1} \| \leq
e^{t_1\| A_{\sigma(0)} \|} \cdot
\prod_{k=1}^{m-1}
e^{(t_{k+1} - t_{k}) \| A_{\sigma(t_{k})} \| } \cdot
e^{(t - t_m) \| A_{\sigma(t_m)} \| } leq
e^{\Lambda t} < e^{\Lambda T_s} \label{eq:state_map_bound}.
\end{align}
It is obvious that the equation above holds in the non-switched case.
Since $q_x(\tau) = q_x(0) = Q(x(0))$ when $0 \leq \tau \leq t~(< T_s)$, 
if follows from \eqref{eq:alpha0_bound} that
\begin{align}
\left\| \int^{t}_0 \Phi(\tau,0)^{-1}
B_{\sigma(\tau)} K_{\sigma(0)} q_x(\tau) d \tau \right\| 
&\leq
\int^{t}_0
\|\Phi(\tau, 0)^{-1} \| \cdot \|B_{\sigma(\tau)} K_{\sigma(0)} q_x(\tau) \| d\tau 
\notag \\
&\leq
\alpha_0 \int^{t}_0 e^{\Lambda \tau} d\tau  \|x(0)\| \notag \\
&\leq
\alpha_0  \frac{e^{\Lambda T_s} - 1}{\Lambda}  \|x(0)\|
= \eta \|x(0)\|.
\label{eq:int_state_map_bound}
\end{align}
Substituting \eqref{eq:state_map_bound} and \eqref{eq:int_state_map_bound}
into \eqref{eq:x0_bound},
we obtain
\begin{equation*}
\|x(0)\| < e^{\Lambda T_s} \|x(t)\| + 
\eta \|x(0)\|.
\end{equation*}
Thus if \eqref{eq:alpha_0_condition} holds, we derive
\eqref{eq:alpha1_bound}.
\end{proof}

Let us next develop an upper bound on the norm of
the error $x(t) - x([t]^-)$ due to sampling.
To this end, 
we show the following proposition:
\begin{proposition}
\label{prop:state_transition_bound}
Let $\Phi(t,0)$ be the state-transition map of 
the switched system \eqref{eq:SLS}
as above. Then
\begin{equation}
\label{eq:Phi_1_diff_bound}
\|\Phi(t,0) - I\| \leq e^{\Lambda t} - 1.
\end{equation}
\end{proposition}

\begin{proof}
Let us first show the case without switching, i.e.,
\begin{equation}
\label{eq:no_switching_matrix_expo_diff_bound}
\|e^{tA_{\sigma(0)}} - I \| \leq e^{\Lambda t} - 1.
\end{equation}
Define the partial sum $S_N$ of $e^{tA_{\sigma(0)}} - I$ by
\begin{equation*}
S_N(t) = 
\sum_{k=0}^N \frac{1}{k!} (tA_{\sigma(0)})^k - I
=
\sum_{k=1}^N \frac{1}{k!} (tA_{\sigma(0)})^k 
\end{equation*}
Then for $t \geq 0$
\begin{align*}
\|S_N(t)\| &\leq
\sum_{k=1}^N \frac{1}{k!} \left(t\|A_{\sigma(0)}\| \right)^k  \\
&=\sum_{k=0}^N \frac{1}{k!} \left(t\|A_{\sigma(0)}\|\right)^k  - 1\\
&\leq \sum_{k=0}^\infty \frac{1}{k!} \left(t\|A_{\sigma(0)}\|\right)^k  - 1\\
&= e^{t\|A_{\sigma(0)}\|} - 1 
\leq e^{\Lambda t} - 1.
\end{align*}
If we let $N \to \infty$, we obtain \eqref{eq:no_switching_matrix_expo_diff_bound}.

We now prove \eqref{eq:Phi_1_diff_bound}
in the switched case.
Let $t_1,t_2,\dots,t_m$ be the switching times in the interval $[0,t)$.
Let $t_0 = 0$ and $t_{m+1} = t$.
Then \eqref{eq:Phi_1_diff_bound} is equivalent to
\begin{align}
\left\|~
\prod_{k=0}^{m} 
e^{(t_{k+1} - t_{k})A_{\sigma(t_{k})}}- I~
\right\| 
\leq e^{\Lambda t} - 1.
\label{eq:matrix_expo_diff_bound}
\end{align}

We have already shown the case $m=0$, i.e., the non-switched case.
The general case follows by induction. For $m \geq 1$,
\begin{align*}
\left\| \prod_{k=0}^{m} 
e^{(t_{k+1} - t_{k})A_{\sigma(t_{k})}}- I 
\right\| 
&\leq
\left\| e^{(t_{m+1} - t_{m})A_{\sigma(t_{m})}} \left(
\prod_{k=0}^{m-1} e^{(t_{k+1} - t_{k})A_{\sigma(t_{k})}}
- I \right) \right\| 
+ 
\| e^{(t_{m+1} - t_{m})A_{\sigma(t_{m})}}- I \| \\
&\leq\| e^{(t_{m+1} - t_{m})A_{\sigma(t_{m})}}\| \cdot
\left\|
\prod_{k=0}^{m-1}
e^{(t_{k+1} - t_{k}) A_{\sigma(t_{k})} } - I
\right\|  
+ 
\| e^{(t_{m+1} - t_{m})A_{\sigma(t_m)}}- I \|.
\end{align*}
Hence if \eqref{eq:matrix_expo_diff_bound} holds with
$m-1$ in place of $m$, then
\begin{align*}
\| e^{(t_{m+1} - t_{m})A_{\sigma(t_{m})}}\| \cdot
\left\|
\prod_{k=0}^{m-1}
e^{(t_{k+1} - t_{k}) A_{\sigma(t_{k})} } - I
\right\| 
&+ 
\| e^{(t_{m+1} - t_{m})A_{\sigma(t_m)}}- I \| \\
&\leq
e^{\Lambda (t_{m+1} - t_{m})}
(e^{\Lambda t_{m}} - 1) + (e^{\Lambda (t_{m+1} - t_{m})} - 1) \\
&= e^{\Lambda t } - 1.
\end{align*}
Thus we obtain \eqref{eq:matrix_expo_diff_bound}.
\end{proof}

\begin{lemma}
\label{lem:beta1_bound}
Consider the switched system \eqref{eq:SLS} with \eqref{eq:control_input}, 
where $\sigma$ has finitely many switching times
on every finite interval.
Define $\beta_1$ by
\begin{align*}
\beta_1 = (e^{\Lambda T_s} - 1) \left( 1 +  \frac{\alpha_0}{\Lambda} \right)
\end{align*}
Then we have
\begin{equation}
\label{eq:beta1_bound}
\|x(t) - x([t]^-) \| < \beta_1 \|x([t]^-)\|
\end{equation}
for all $t \geq 0$ with 
$x([t]^-) \in
\overline{\mathcal{E}}_P(R)$.
\end{lemma}
\begin{proof}
As in the proof of Lemma \ref{lem:alpha1_bound}, it suffices to prove
\eqref{eq:beta1_bound} for all 
$x(0) \in \overline{\mathcal{E}}_P(R)$ and 
$t \in [0,T_s)$.

By \eqref{eq:state_t}, we obtain
\begin{align*}
x(t) - x(0) = 
(\Phi(t,0) - I) x(0)  + \int^{t}_0 \Phi(t,\tau)
B_{\sigma(\tau)} K_{\sigma(0)} q_x(\tau) d \tau.
\end{align*}
This leads to
\begin{align}
\|x(t) - x(0)\| \leq
\| \Phi(t,0) - I \| \cdot \|x(0) \| + \left\| \int^{t}_0 \Phi(t,\tau)
B_{\sigma(\tau)} K_{\sigma(0)} q_x(\tau) d \tau \right\|.
\label{eq:xt_x0_dif_bound}
\end{align}

Proposition \ref{prop:state_transition_bound} provides
the following upper bound on the first term of the right side of 
\eqref{eq:xt_x0_dif_bound}:
\begin{align}
\label{eq:first_term_upper}
\|\Phi(t,0) - I\| \leq
e^{\Lambda t} - 1 < e^{\Lambda T_s} - 1.
\end{align}

Since a
calculation similar to \eqref{eq:state_map_bound} shows that
$\| \Phi (t,\tau) \| \leq e^{\Lambda(t-\tau)}$. 
Hence as in \eqref{eq:int_state_map_bound},
\begin{align}
\left\| 
\int^{t}_0 \Phi(t,\tau)
B_{\sigma(\tau)} K_{\sigma(0)} q_x(\tau) d \tau \right\|
&\leq
\alpha_0 \frac{e^{\Lambda T_s} - 1}{\Lambda} \|x(0)\|.
\label{eq:int_Phit_tau}
\end{align}

We obtain \eqref{eq:beta1_bound} by
combining \eqref{eq:first_term_upper}
with \eqref{eq:int_Phit_tau}.
\end{proof}

Similarly to \eqref{eq:alpha0_bound},
to each $p,q \in \mathcal{P}$ with $p\not=q$, 
there correspond a positive number
$\gamma_0(p,q)$ such that
\begin{align}
\label{eq:beta2_bound}
\|PB_pK_q(Q(x) - x)\| &\leq \gamma_0(p,q) \|x\|
\end{align}
for $x \in
\overline{\mathcal{E}}_P(R)$.

We are now in a position to obtain upper bounds on
the norm of  
the error 
$q_x(t) - x(t)$ due to
sampling and quantization by using
the original state $x(t)$.
\begin{theorem}
\label{thm:alpha_beta_bound}
Consider the switched system \eqref{eq:SLS} with \eqref{eq:control_input}, 
where $\sigma$ has finitely many switching times
on every finite interval.
Define $\alpha_1$ and $\beta_1$ as in Lemmas 
\ref{lem:alpha1_bound} and \ref{lem:beta1_bound}.
If $\beta(p,q)$ are defined by
\begin{align*}
\gamma(p,q) &= \alpha_1 (\beta_1\|PB_pK_q\|  + \gamma_0(p,q)),
\end{align*}
then $\beta(p,q)$ satisfy
\begin{align}
\|PB_pK_q(q_x(t) - x(t))\| &< \gamma(p,q)\|x(t)\| \label{eq:beta_bound}
\end{align}
for all $t \geq 0$ with 
$x([t]^-) \in
\overline{\mathcal{E}}_P(R)$.
\end{theorem}
\begin{proof}
It follows from \eqref{eq:beta1_bound} and \eqref{eq:beta2_bound} that
\begin{align*}
\|PB_pK_q(q_x(t) - x(t))\| 
&\leq 
\|PB_pK_q(q_x(t) - x([t]^-)) \|  + \|PB_pK_q\|\cdot\| x([t]^-) - x(t) \| \\
&<
(\beta_1\|PB_pK_q\| + \gamma_0(p,q)) \|x([t]^-)\| \\
&< \alpha_1 (\beta_1\|PB_pK_q\| + \gamma_0(p,q))\|x(t)\|.
\end{align*}
Thus the second inequality \eqref{eq:beta_bound} holds.
\end{proof}

An upper bound on $\dot V_{p,q}$
can be obtained as follows.


Since $\dot V_{p,q}$ satisfies
\begin{align*}
\dot V_{p,q}(x(t),q_x(t)) 
=
2x(t)^{\top}P(A_p + B_pK_q)x(t) +
2x(t)^{\top}PB_pK_q(q_x(t) - x(t)),
\end{align*}
we see from \eqref{eq:beta_bound} that
\begin{equation}
\label{eq:second_bound_Vpq}
\dot V_{p,q}(x(t),q_x(t)) \leq
2(\|P(A_p + B_pK_q) \| + \gamma(p,q) ) \|x(t)\|^2
\end{equation}
for all $t \geq 0$ with 
$x([t]^-) \in
\overline{\mathcal{E}}_P(R)$.
Fast sampling and fine quantization make
the upper bound \eqref{eq:second_bound_Vpq} small.

Define $D$ by
\begin{align}
D = 
2\max_{p \not= q} 
(\|P(A_p + B_pK_q) \| + \gamma(p,q)).
\label{eq:D_def}
\end{align}
Then we obtain
\begin{align}
\label{eq:dotVpq_bound}
\dot V_{p,q}(x(t),q_x(t)) \leq D \|x(t)\|^2
\end{align}
for $p,q \in \mathcal{P}$ with $p \not= q$ and
for $t \geq 0$ with
$x(t)\in
\overline{\mathcal{E}}_P(R)~\backslash~\underline{\mathcal{E}}_P(r)$
and $x([t]^-) \in
\overline{\mathcal{E}}_P(R)$.



\subsection{Stability analysis with total mismatch time}
Let us analyze the stability of 
the switched system \eqref{eq:SLS} with \eqref{eq:control_input} by the two
upper bounds \eqref{eq:dotVp_bound} and 
\eqref{eq:dotVpq_bound} of $\dot V$.
Note that
the former bound \eqref{eq:dotVp_bound} 
is for the case $\sigma(t)=\sigma([t]^-)$, while
the latter \eqref{eq:dotVpq_bound} for the case $\sigma(t)\not=\sigma([t]^-)$.
It is therefore useful to define the following characterization of the 
switching signal:
\begin{definition}
For $\tau_1 > \tau_2 \geq 0$, 
we define the total mismatch time $\mu(\tau_1,\tau_2)$ by
\begin{equation}
\label{eq:mu_def_same}
\mu(\tau_1, \tau_2) = \text{{\em the length
of~}} \{\tau \in [\tau_2,\tau_1):~ \sigma(\tau) \not = \sigma([\tau]^-)\}
\end{equation}
\end{definition}
More explicitly, 
the length of an interval means its Lebesgue measure.
We shall not, however, use any measure theory because
$\sigma$ has only finitely many discontinuities
on every interval.

Define $C_P$ and $D_P$ by
\begin{equation*}
C_P = \frac{C}{\lambda_{\max}(P)}, \quad
D_P = \frac{D}{\lambda_{\min}(P)}.
\end{equation*} 
First we study the state behavior when it is outside of
$\underline{\mathcal{E}}_P(r)$.
The following lemma suggests that every trajectory with its initial state in
$\Int(\overline{\mathcal{E}}_P(R))$ goes into $\underline{\mathcal{E}}_P(r)$
if the total mismatch time $\mu$ is sufficiently small; see Fig. \ref{fig:Lemma_Explain}.
\begin{lemma}
\label{lem:tor0}
Let Assumptions \ref{ass:system}, \ref{ass:quantization_near_origin}, and
\ref{ass:subsystem_QA}
hold, and let
$L \geq 0$ satisfy
\begin{equation}
\label{eq:L_inequality1}
L < \frac{C_P}{C_P + D_P}.
\end{equation}
If $\mu(t,0)$ achieves 
\begin{equation}
\label{eq:mu_inequality1}
\mu(t,0) \leq Lt
\end{equation}
for $t > 0$, then 
there exists $T_{r} \geq 0$ such that for every
$x(0) \in \Int(\overline{\mathcal{E}}_P(R))$
and $\sigma(0) \in \mathcal{P}$,
$x(T_{r}) \in \underline{\mathcal{E}}_P(r)$ and
$x (t) \in \Int (\overline{\mathcal{E}}_P(R))$ for all $t \in [0, T_{r}]$.
\end{lemma}
\begin{proof}
First we show that the trajectory $x$ does not leave
$\Int (\overline{\mathcal{E}}_P(R))$ without belonging 
to $\underline{\mathcal{E}}_P(r)$. That is, 
there does not exist $T_{R} > 0$
such that 
\begin{gather}
\label{eq:x(TR0)_outside}
x(T_{R}) \in \partial \overline{\mathcal{E}}_P(R) \\
\label{eq:x(t)_0_TR0}
x(t) \in 
\Int (\overline{\mathcal{E}}_P(R))~\backslash~\underline{\mathcal{E}}_P(r)
\qquad (0 \leq t < T_{R}).
\end{gather}

Assume, to reach a contradiction, \eqref{eq:x(TR0)_outside} and
\eqref{eq:x(t)_0_TR0} hold for some $T_{R} > 0$.
Recall that
\begin{equation*}
\lambda_{\min}(P) \|x\|^2 \leq V(x) = x^{\top} Px \leq 
\lambda_{\max}(P) \|x\|^2 
\end{equation*}
for $x \in \mathbb{R}^{\sf n}$.
It follows from
\eqref{eq:dotVp_bound} and \eqref{eq:dotVpq_bound} that
\begin{align}
\label{eq:VpVpq_bound_by_V}
\begin{array}{c}
\dot V_p(x(t),q_x(t)) \leq -C_PV(x(t)) \\[4pt]
\dot V_{p,q}(x(t),q_x(t)) \leq D_P V(x(t)).
\end{array}
\end{align}

By \eqref{eq:VpVpq_bound_by_V}, a successive calculation at each switching time
shows that
\begin{align}
\label{eq:VTR0_bound_by_V0}
V(x(T_{R})) 
\leq 
\exp \big(D_P \mu(T_{R},0)-C_P(T_{R} - \mu(T_{R},0))\big) 
V(x(0)).
\end{align}
Since \eqref{eq:mu_inequality1} gives
\begin{align}
D_P \mu(t,0)-C_P(t - \mu(t,0))
\leq
\left(
\left(
C_P + 
D_P
\right) L - 
C_P
\right)t
\label{eq:CD_bound0}
\end{align}
for $t > 0$,
it follows from \eqref{eq:L_inequality1} and 
$x(0) \in \Int(\overline{\mathcal{E}}_P(R))$ that
\begin{align*}
V(x(T_{R})) < V(x(0)) < R^2 \lambda_{\max} (P).
\end{align*}
However, \eqref{eq:x(TR0)_outside} shows that
$
V(x(T_{R})) 
= R^2 \lambda_{\max} (P),
$
and we have a contradiction.

Let us next prove that
$x(T_{r}) \in \underline{\mathcal{E}}_P(r)$ for some
$T_{r} \geq 0$.

Suppose $x(t) \not\in \underline{\mathcal{E}}_P(r)$ for 
all $t \geq 0$. Then since the discussion above shows that 
$x(t) \in \Int (
\overline{\mathcal{E}}_P(R))~\backslash~\underline{\mathcal{E}}_P(r)$
for $t \geq 0$, we obtain \eqref{eq:VTR0_bound_by_V0} with arbitrary
$t \geq 0$ in place of $T_{R}$.
Hence \eqref{eq:L_inequality1} and \eqref{eq:CD_bound0} show that
$V(x(t)) \to 0$ as $t \to \infty$.
However this contradicts
$x(t) \not\in \underline{\mathcal{E}}_P(r)$, i.e.,
$V(x(t)) > r^2 \lambda_{\min}(P)$. Thus there exists 
$T_{r} \geq 0$ such that 
$x(T_{r}) \in \underline{\mathcal{E}}_P(r)$.
\end{proof}

From the next result, we see that
the trajectory leaves $\underline{\mathcal{E}}_P(r)$ only if
a switch occurs between sampling times.
\begin{lemma}
\label{lem:cross_to_out}
Let Assumptions \ref{ass:system}, \ref{ass:quantization_near_origin}, and
\ref{ass:subsystem_QA} hold.
Let the trajectory $x(t)$ leave
$\underline{\mathcal{E}}_P(r)$ when $t = T_0$.
More precisely, 
there exists $\delta > 0$ such that
\begin{equation}
\label{eq:x(T_0)_def}
x(T_0) \in \partial \underline{\mathcal{E}}_P(r),~
x(T_0+\varepsilon) \not\in \underline{\mathcal{E}}_P(r)
~~~~ (0 <  \varepsilon < \delta).
\end{equation}
Then $\sigma(T_0) \not= \sigma([T_0]^-)$.
\end{lemma}

\begin{proof}
This immediately follows from the fact that 
$\underline{\mathcal{E}}_P(r)$ is an invariant set
if the mode mismatch does not happen.
\end{proof}

Lemma \ref{lem:r_to_ar} below shows that if 
the trajectory enters into $\underline{\mathcal{E}}_P(r)$,
it keeps roaming a slightly larger ellipsoid than 
$\underline{\mathcal{E}}_P(r)$; see Fig. \ref{fig:Lemma_Explain}.  
\begin{lemma}
\label{lem:r_to_ar}
Let Assumptions \ref{ass:system}, \ref{ass:quantization_near_origin}, and
\ref{ass:subsystem_QA} hold.
Suppose that $T_0 \geq 0$ is the first time at which
$x(t)$ leaves $\underline{\mathcal{E}}_P(r)$.
Let $a > 1$ satisfy 
\begin{equation}
\label{eq:a_cond}
a^2r^2 \lambda_{\min}(P)
< R^2\lambda_{\max}(P)
\end{equation}
and define $b(a)$ by
\begin{equation}
\label{eq:ba_bound}
b(a) =
\frac{2\log a}{C_P + D_P}.
\end{equation}
Pick $L \geq 0$ with \eqref{eq:L_inequality1} and
suppose that $\mu(t,T_0)$ satisfies
\begin{equation}
\label{eq:mu_b_L_condition}
\mu(t,T_0) \leq b(a) + L(t-T_0)
\end{equation}
for all $t > T_0$. Then
there exists $T_{1} > T_0$ such that 
for every $\sigma(T_0) \in \mathcal{P}$,
$x(T_1) \in \underline{\mathcal{E}}_P(r)$ and
$x (t) \in \Int(\underline{\mathcal{E}}_P(ar))$ 
for $t \in [T_0, T_1]$.
\end{lemma}
\begin{proof}
By \eqref{eq:mu_b_L_condition},
as long as $x(t)\in
\overline{\mathcal{E}}_P(R)~\backslash~\underline{\mathcal{E}}_P(r)$,
$V(x(t))$ satisfies
\begin{align}
V(x(t)) \leq 
\exp\big( 
\left(
\left(
C_P+ D_P \right)L - 
C_P
\right) (t - T_0)
\big) 
\exp\big(
\left(C_P  + D_P\right)b(a)
\big) V(x(T_0) ) \label{eq:Vt_VT0}
\end{align}
for $t \geq T_0$.
On the other hand,
since $x(T_0) \in \partial \underline{\mathcal{E}}_P(r) $,
it follows from \eqref{eq:ba_bound} that
\begin{align}
\label{eq:Lyapunov_bound_ar0}
\exp\big(
\left(C_P + D_P\right)b(a)
\big) V(x(T_0) ) 
\leq a^2 \lambda_{\min}(P).
\end{align}
Since \eqref{eq:a_cond} holds if and only if
$
\underline{\mathcal{E}}_P(ar) 
\subset \Int (\overline{\mathcal{E}}_P(R)),
$
as in the proof of Lemma \ref{lem:tor0}, \eqref{eq:Vt_VT0} leads to
$x(T_1) \in \underline{\mathcal{E}}_P(r)$ for some
$T_1 > T_0$.
Substituting \eqref{eq:Lyapunov_bound_ar0} into \eqref{eq:Vt_VT0},
we also obtain 
$V(x(t)) < a^2\lambda_{\min}(P)$ for $t \geq T_0$. 
Thus
$x(t) \in \Int (\underline{\mathcal{E}}_P(ar))$ 
for $t \in [T_0, T_1]$.
\end{proof}

 \begin{figure}[t]
 \centering
 \includegraphics[width = 7cm, clip]{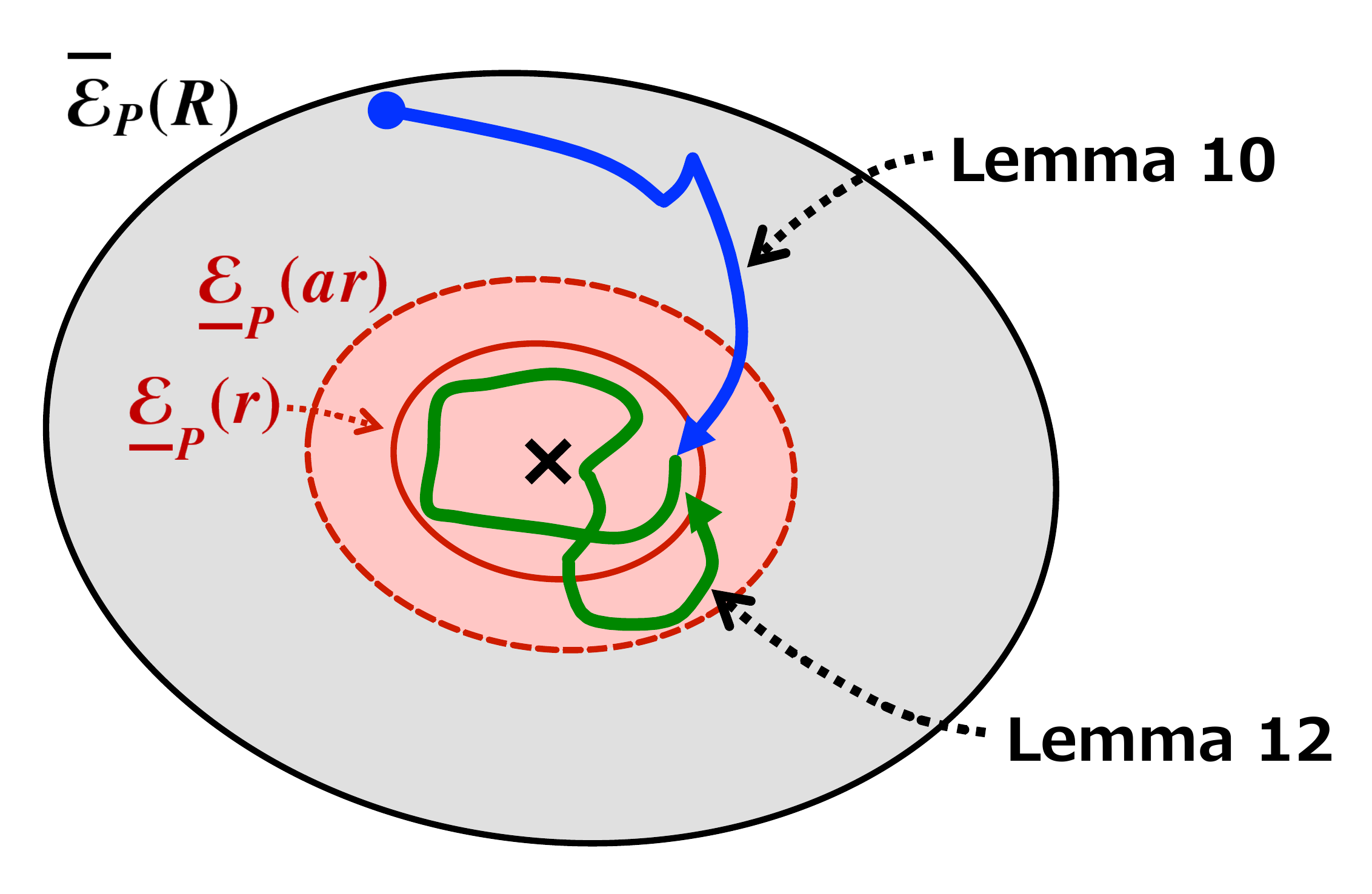}
 \caption{Behavior of trajectory}
 \label{fig:Lemma_Explain}
 \end{figure}

Referring to 
Lemmas \ref{lem:tor0}, \ref{lem:cross_to_out}, and \ref{lem:r_to_ar},
we immediately derive the following result:
\begin{theorem}
\label{thm:main_thm}
Let Assumptions \ref{ass:system}, \ref{ass:quantization_near_origin}, and
\ref{ass:subsystem_QA} hold.
Let $L$, $a$, and $b(a)$ be as in Lemmas \ref{lem:tor0} and \ref{lem:r_to_ar}.
Suppose that $\mu$ in \eqref{eq:mu_def_same} satisfies
\eqref{eq:mu_inequality1} for $t > 0$ and 
\eqref{eq:mu_b_L_condition} for $t > T_0$ with 
$\sigma(T_0) \not= \sigma([T_0]^-)$.

If $x(0) \in \Int(\overline{\mathcal{E}}_P(R))$, then
every trajectory $x$ of the switched system \eqref{eq:SLS} 
with \eqref{eq:control_input} satisfies
$x(t) \in \Int(\overline{\mathcal{E}}_P(R))$ for $t \geq 0$, and furthermore
there exists $T_{r} \geq 0$ such that 
$x(t) \in \Int(\underline{\mathcal{E}}_P(ar))$ 
for $t \geq T_{r}$.
\end{theorem}

\begin{remark}
In this section, we have studied the stability analysis of the switched system
by using the total mismatch time of the modes 
between the plant and the feedback gain.
If the mismatch {\em does} occur, 
the closed-loop system may be {\em unstable},
If {\em not}, it is {\em stable}.
Our proposed method is therefore similar to that in \cite{Zhai2001}, where
the stability analysis of switched systems with stable and unstable subsystems 
is discussed with the aid of 
the total activation time ratio between stable subsystems and unstable ones.
In \cite{Zhai2001}, the average dwell time
introduced by \cite{Hespanha1999CDC}
is also required to be sufficiently large.
However, such a requirement is not necessary here because we use a common Lyapunov function.
\end{remark}

\section{Reduction to a Dwell-Time Condition}
In the preceding section, we have derived a sufficient condition 
on the total mismatch time $\mu$ for the stabilization of
the switched system with limited information.
However it may be difficult to check whether $\mu$ satisfies
\eqref{eq:mu_inequality1} and \eqref{eq:mu_b_L_condition}.
In this section, we will briefly show that these conditions \eqref{eq:mu_inequality1} and \eqref{eq:mu_b_L_condition}
can be achieved for
switching signals with a certain dwell time property.


To proceed, we recall the definition of dwell time.
If the switching signal $\sigma$ has an interval between
consecutive discontinuities no smaller than $T_d > 0$, and further if $\sigma$ 
has no discontinuities in $[0, T_d)$, then we call $\sigma$ 
a \textit{switching signal with dwell time $T_d$}.


\begin{proposition}
\label{thm:mu_upper_bound}
Fix $n \in \mathbb{N}$.
For every $\sigma$ with dwell time $nT_s$, $\mu$ in \eqref{eq:mu_def_same}
satisfies
\begin{equation}
\label{eq:mu_0T} 
\mu(t,0) < \frac{1}{n}t \qquad (t > 0).
\end{equation}
Furthermore, if $\sigma(T_0) \not= \sigma([T_0]^-)$, then 
\begin{equation} 
\label{eq:mu_T0T}
\mu(t,T_0) < T_s + \frac{1}{n}(t-T_0) 
\qquad (t > T_0).
\end{equation}
\end{proposition}

\begin{proof}
The proof includes a lengthy but routine calculation;
see Appendix A.1.
\end{proof}

Theorem \ref{thm:main_thm} and 
Proposition \ref{thm:mu_upper_bound} can be 
combined in the following way:
\begin{theorem}
\label{thm:dwell_sample_switching}
Let Assumptions \ref{ass:system}, \ref{ass:quantization_near_origin}, and
\ref{ass:subsystem_QA} hold.
Let $n \in \mathbb{N}$ satisfy $n \geq 1+D_P/C_P$.
Define 
\begin{equation*}
a = \exp 
\left(
\frac{T_s(C_P + D_P)}{2}
\right),
\end{equation*}
and suppose that $a$ satisfies \eqref{eq:a_cond}.
If $x(0) \in \Int(\overline{\mathcal{E}}_P(R))$
and if the dwell time of $\sigma$ is $nT_s$, then
every trajectory $x(t)$ of the switched system \eqref{eq:SLS} 
with \eqref{eq:control_input} satisfies
$x(t) \in \Int(\overline{\mathcal{E}}_P(R))$ for $t \geq 0$, and furthermore
there exists $T_{r} \geq 0$ such that 
$x(t) \in \Int(\underline{\mathcal{E}}_P(ar))$ for $t \geq T_{r}$.
\end{theorem}
\begin{proof}
If $n$ and $a$ are defined as above, 
Proposition \ref{thm:mu_upper_bound} shows that
$\mu$ satisfies
\eqref{eq:mu_inequality1} and
\eqref{eq:mu_b_L_condition} for every $\sigma$ with dwell time $nT_s$.
Hence the conclusion of Theorem \ref{thm:main_thm} holds.
\end{proof}

The next result implies that the upper bounds obtained in 
Proposition \ref{thm:mu_upper_bound} are close to the supremum if
the sampling period is enough small.
\begin{proposition}
\label{thm:mu_lower_bound}
Fix $\varepsilon > 0$ and $n \in \mathbb{N}$.
For any $T \geq 0$,
there exist $\sigma$ with dwell time $nT_s$ and $t \geq T$ such that
\begin{align*}
\mu(0,t) \geq \frac{1}{n}t - \left( \frac{T_s}{n} + \varepsilon \right).
\end{align*}
Furthermore, for any $T \geq 0$,
there exist $\sigma$ with dwell time $nT_s$, $T_0 \geq 0$ with
$\sigma(T_0) \not= \sigma([T_0]^-)$, and
$t \geq T_0 + T$ such that 
\begin{align}
\label{eq:lowerbound2}
\mu(T_0,t) \geq T_s + \frac{1}{n}(t-T_0) - \left( \frac{T_s}{n} + \varepsilon \right).
\end{align}
\end{proposition}
\begin{proof}
This is again a routine calculation; see Appendix A.2.
\end{proof}

\begin{corollary}
\label{cor:small_dwelltime}
There exist a switching signal with dwell time $T_s$ such
that $\mu(0,t) \approx t$ for a sufficiently large $t>0$.
\end{corollary}
This is of course the expected result.
Corollary \ref{cor:small_dwelltime} shows that
if the dwell time does not exceed the sampling period,
the information about the switching signals is meaningless to
stabilize the switched system.

 \section{Numerical Example}
Consider the switched system with the following two modes:
\begin{gather*}
A_1 = 
\frac{1}{6}
\begin{bmatrix}
1 & -2\\ -3 & 2
\end{bmatrix}, \quad
B_1 = 
\frac{1}{6}
\begin{bmatrix}
-4 \\ 3
\end{bmatrix}, \\
A_2 = 
\begin{bmatrix}
1 & -5\\ 1 & 2
\end{bmatrix}, \quad
B_2 = 
\begin{bmatrix}
1 \\ -1
\end{bmatrix}.
\end{gather*}
The state feedback gains $K_1$ and $K_2$ are given by
\begin{equation}
\label{eq:state_feedback_gain_Ex}
K_1 = \begin{bmatrix}
1.38 &  -1.86
\end{bmatrix},\quad
K_2 = \begin{bmatrix}
-2.80 &  3.77
\end{bmatrix}.
\end{equation}
We computed the above regulator gains by minimizing the cost
\begin{align*}
\int^{\infty}_{0} \big( x(t)^{\top}x(t) + u(t)^2 \big) dt.
\end{align*}
Note that both $A_1 - B_1K_2$ and $A_2 - B_2K_1$ are not Hurwitz:
$A_1 - B_1K_2$ has one unstable eigenvalue $4.4538$ and
$A_2 - B_2K_1$ has two unstable eigenvalues $1.4091$ and
$4.7750$.

The sampling period $T_s$ is given by $T_s = 0.025$, 
and we used the following logarithm quantizer:
Let the state $x$ be 
$x = [x_1~~x_2]^{\top}$.
For a nonnegative integer $n$, 
the quantized state 
$Q(x) = [Q_1(x_1)~~Q_2(x_2)]^{\top}$
is defined by
\begin{equation*}
Q_i(x_i) =
\begin{cases}
 \frac{-\xi_0 (\kappa^n + \kappa^{n+1})}{2}
&  (-\xi_0 \kappa^{n+1} \leq x_i < -\xi_0\kappa^n)\\ 
0
& (-\xi_0 \leq x_i \leq \xi_0) \\
\frac{\xi_0 (\kappa^n + \kappa^{n+1})}{2}
& (\xi_0\kappa^n < x \leq \xi_0\kappa^{n+1}),
\end{cases}
\end{equation*}
where $\xi_0 = 0.08$ and $\kappa = 1.2$.

Set $C = 1$, $R=68.6$, and $r=0.175$ in Assumption~\ref{ass:subsystem_QA}.
In conjunction with 
the scheduling function having the revisitation property in \cite{Liberzon2004},
the randomized algorithm in \cite{Ishii2004} gave
\begin{equation*}
P = 
\begin{bmatrix}
2.9171 & 0.3489 \\ 
0.3489 & 3.6256
\end{bmatrix}.
\end{equation*}
In the randomized algorithm, 
we used $10^7$ samples in state for each run,
and five samples in time for each sampled state.
We stopped the algorithm when there is no update for an entire run.

Since we obtain $D = 61.02$ in \eqref{eq:dotVpq_bound} from the data above,
the resulting $n$ and $a$ in Theorem \ref{thm:dwell_sample_switching} are
$n=84$ and $a=1.321$.

A time response ($0 \leq t \leq 30$) was calculated for $x(0) = [-60~~50]^{\top}$ and 
$\sigma(0)=1$.
Fig.~\ref{fig:LQ11fig} depicts the state trajectory $x$ 
of the switched system~\eqref{eq:SLS} with dwell time $84T_s = 2.1$.
The blue line indicates that the feedback gain designed for the active subsystem
was used, i.e.,
\begin{equation*}
(A_{\sigma(t)},B_{\sigma(t)},K_{\sigma([t]^-)})
=
(A_{1},B_{1},K_{1})
~\text{or}~
(A_{2},B_{2},K_{2}).
\end{equation*}
The red line shows that a switch led to the mismatch of the modes between
the plant and the feedback gain, i.e.,
\begin{equation*}
(A_{\sigma(t)},B_{\sigma(t)},K_{\sigma([t]^-)})
=
(A_{1},B_{1},K_{2})
~\text{or}~
(A_{2},B_{2},K_{1}).
\end{equation*}
The black lines in Fig.~\ref{fig:LQ11fig}
represent $\overline{\mathcal{E}}_P(R)$ and
$\underline{\mathcal{E}}_P(ar)$, respectively.

Here we see two conservative results: the dwell time $84T_s$
and the ellipsoid $\underline{\mathcal{E}}_P(ar)$ in Fig.~\ref{fig:LQ11fig3}.
Since we evaluate the increase and decrease of the Lyapunov function 
only by \eqref{eq:dotVpq_bound} and \eqref{eq:dotVp_bound},
the switching condition for stabilization become conservative.
In particular, we need to refine the upper bound \eqref{eq:dotVpq_bound}
in the case of mode mismatch.

As regards the ellipsoid $\underline{\mathcal{E}}_P(ar)$,
the trajectory in Fig.~\ref{fig:LQ11fig3} remains in
a smaller neighborhood of the origin. 
This conservative result is also due to 
the upper bound \eqref{eq:dotVpq_bound}; see Lemma \ref{lem:r_to_ar}.
Another reason is 
the nonlinearity of quantization
and this is observed in the non-switched case \cite{Ishii2004} as well.
If we use multiple Lyapunov functions instead
of a common Lyapunov function, we may reduce this conservativeness.
Details, however, are more involved, so this extension is a subject for future research.

\begin{figure}[b]
\centering
\subcaptionbox{Region $[-70, 10] \times [-10, 60]$ \label{fig:LQ11fig1}}
{\includegraphics[width = 6.5cm,clip]{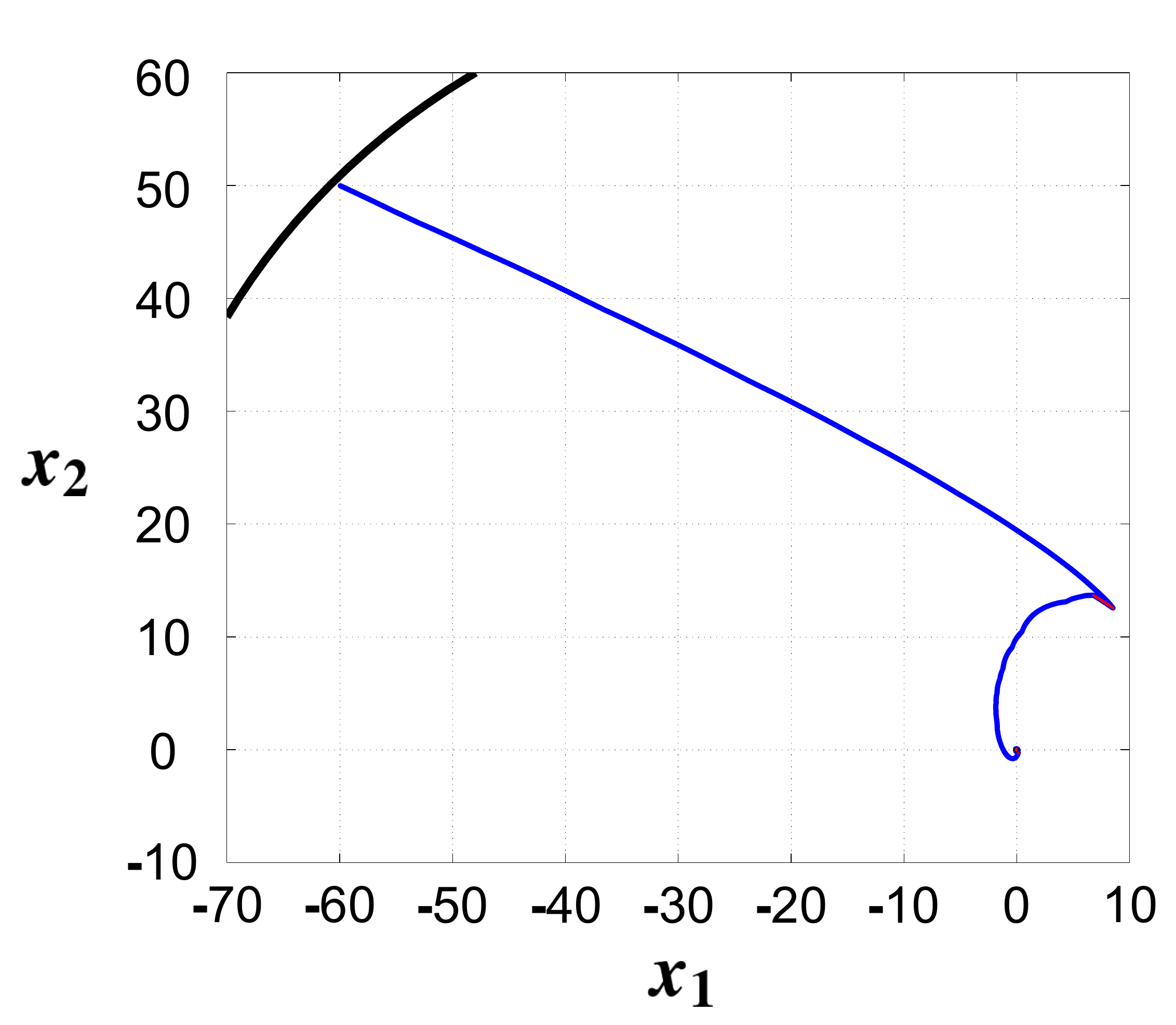}}

\subcaptionbox{Region $[-0.5, 0.3] \times [-0.8, 0.3]$ \label{fig:LQ11fig3}}
{\includegraphics[width = 5cm,clip]{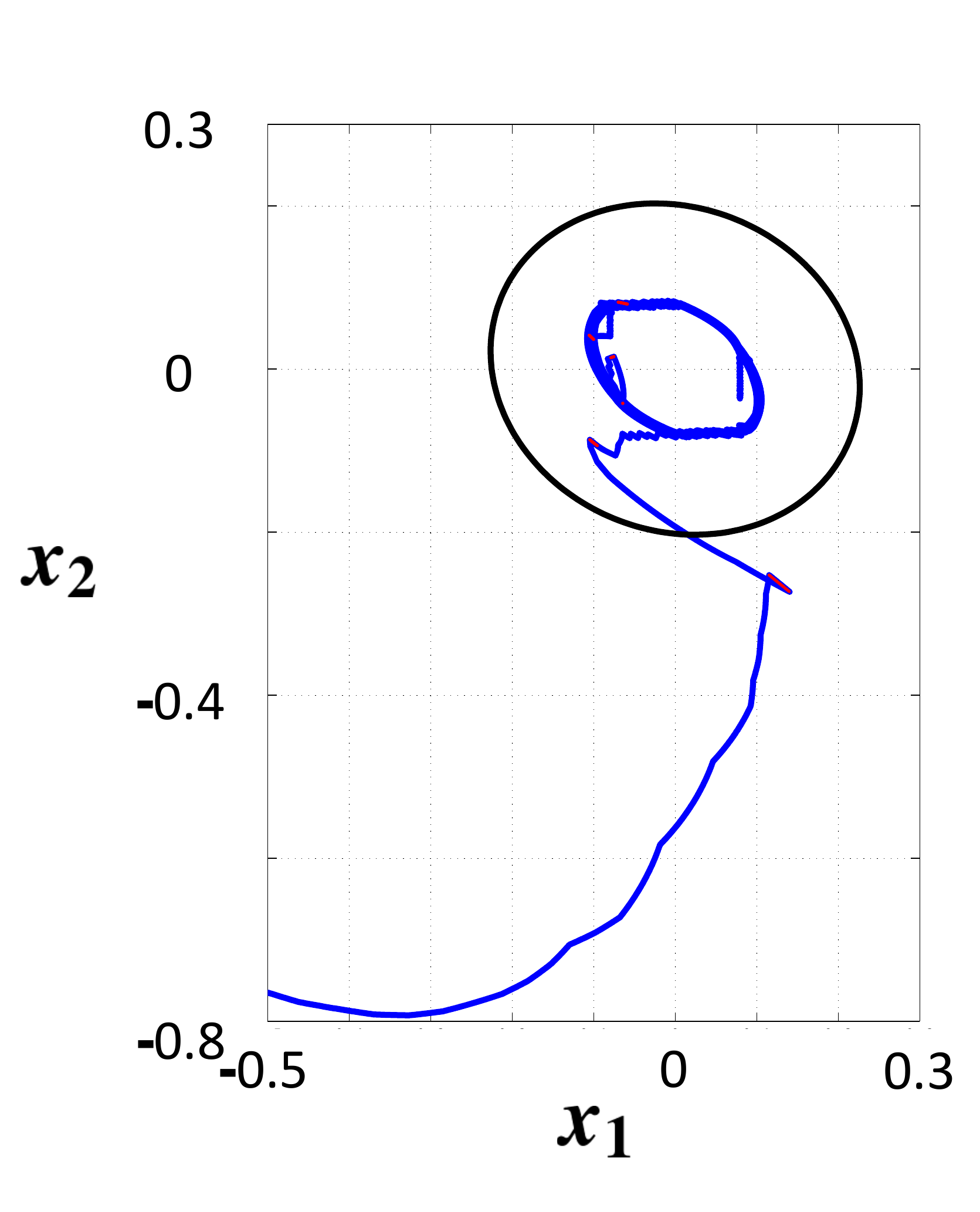}}
%
\caption{The trajectory $x$ with $x(0) = [-60~~50]^{\top}$ and $\sigma(0)=1$
\label{fig:LQ11fig}}
\end{figure}

\section{Concluding Remarks}
We have analyzed the stability of a sampled-data switched systems with 
a memoryless quantizer.
The proposed method uses a common Lyapunov function 
computed efficiently by a randomized algorithm.
The common Lyapunov function leads to the switching conditions on the total mismatch time
for quantized state feedback stabilization.
We have also examined the relationship between the mismatch time
and the dwell time of the switching signal.
Future works will focus on improving the upper bound on the time
derivative of the Lyapunov function in the mismatch case; and
analyzing the stability by using
multiple Lyapunov functions and an average dwell time property.

\appendix
\section{Bound on Total Mismatch Time}    
\subsection{Proof of Proposition \ref{thm:mu_upper_bound}} 
Let us first prove \eqref{eq:mu_0T}.
It is clear that $\mu = 0$ if $\sigma$ has no discontinuities in the interval $(0,t)$.
So suppose $t_1,\dots,t_m$ are switching times in $(0,t)$.
Then since $t \geq mnT_s$, we obtain
\begin{equation*}
\mu(0,t) \leq \sum_{k=1}^{m} ([t_k]^-+T_s - t_k)
< mT_s \leq \frac{1}{n}t.
\end{equation*}
Hence \eqref{eq:mu_0T} holds.

Next we show \eqref{eq:mu_T0T}.
Since $\sigma(T_0) \not= \sigma([T_0]^-)$ and since
the dwell time is $nT_s \geq T_s$, it follows that
$\sigma$ has precisely one discontinuity in the interval $([T_0]^-, T_0]$.
Let us denote the switching time by $t_0$.

Suppose no switches occur in the interval $(T_0, t)$.
Then
\begin{equation*}
\mu(T_0, t) \leq [T_0]^- +T_s - T_0 < T_s,
\end{equation*}
and hence \eqref{eq:mu_T0T} holds.

Suppose $m$ switches occurs in the interval $(T_0, t)$, and let
$t_1,\dots,t_m$ be the switching times.
Define $\xi_k$ by
\begin{equation}
\label{eq:xi_def}
\xi_k = (t_{k+1} - t_k ) - nT_s
\end{equation}
for $k=0,\dots,m-1$.
Our dwell-time assumption shows that $\xi_k \geq 0$.
We also have
\begin{align}
t - T_0 &= 
(t - t_m) + \sum_{k=0}^{m-1}(t_{k+1} - t_{k}) - (T_0 - t_0) \notag \\
&=
(t - t_m) + \sum_{k=0}^{m-1}(\xi_{k} + nT_s) - (T_0 - t_0) \notag\\
&=
mnT_s + (t-t_m) + \sum_{k=0}^{m-1}\xi_{k} - (T_0 -t_0).
\label{eq:t-T_0}
\end{align}
It is best to split the argument two cases.

First we study the case
\begin{equation}
\label{eq:case1}
(t-t_m) + \sum_{k=0}^{m-1}\xi_k \geq T_0 -t_0.
\end{equation}
Combining \eqref{eq:case1} with \eqref{eq:t-T_0}, 
we obtain $t-T_0 \geq mnT_s$, and hence
\begin{align*}
\mu(T_0, t) &\leq  ([T_0]^-+T_s - T_0)
+ \sum_{k=1}^{m} ([t_k]^-+T_s - t_k) \\
&< (m+1)T_s \leq T_s + \frac{1}{n}(t-T_0),
\end{align*}
which is a desired inequality \eqref{eq:mu_T0T}.

Let us next investigate the case
\begin{equation}
\label{eq:case2}
(t-t_m) + \sum_{k=0}^{m-1}\xi_k < T_0 -t_0.
\end{equation}
Since $\sum_{k = 0}^{m-1}\xi_{k} < T_0$ by \eqref{eq:case2},
each switching time $t_k$ ($k=1,\dots,m$) satisfies
\begin{align*}
t_k &= (t_k - t_{k-1}) + \dots + (t_1 - t_0) \\
&= \sum_{\ell = 0}^{k-1} (\xi_{\ell} + nT_s)
\leq  \sum_{\ell = 0}^{m-1}\xi_{\ell} + knT_s 
< T_0 + knT_s.
\end{align*}
In conjunction with the assumption of the dwell time, 
this leads to
\begin{equation}
\label{eq:t_k_bound}
t_0+knT_s \leq t_k < T_0+knT_s
\end{equation}
for $k=1,\dots,m$. 
Since 
\begin{equation}
\label{eq:t0-T0}
[t_0]^- = [T_0]^- < t_0 \leq T_0 < [T_0]^- + T_s, 
\end{equation}
\eqref{eq:t_k_bound} shows that
\begin{align}
\mu([t_k]^-, [t_{k+1}]^-) 
&=
\mu([t_k]^-, [t_k]^-+T_s) \notag \\
&\leq
\mu(t_0, [t_0]^-+T_s) \notag \\
&=
[t_0]^- +T_s - t_0
\label{eq:t_k-t_k+1_bound}
\end{align}
for $k = 1,\dots,m-1$.

Moreover, \eqref{eq:case2} and \eqref{eq:xi_def} give
\begin{equation*}
t < (T_0 - t_0) + t_m -\sum_{k=0}^{m-1} \xi_{k}
= T_0 + mnT_s.
\end{equation*}
If we combine this with
$t > t_m$ and \eqref{eq:t_k_bound},
we see that
\begin{equation*}
t_0 + mnT_s \leq t_m < t < T_0 + mnT_s.
\end{equation*}
Hence 
\begin{equation}
\label{eq:t-t_m_bound}
\mu([t_m]^-,t) = t - t_m< T_0 - t_0.
\end{equation}

Since $t- t_m > 0$ and $\xi_k\geq 0$,
it follows from \eqref{eq:t-T_0} that
\begin{align*}
m < \frac{t-t_0}{nT_s}.
\end{align*}
By \eqref{eq:t_k-t_k+1_bound} and
\eqref{eq:t-t_m_bound}, we have
\begin{align}
\mu(T_0,t) 
&= \mu(T_0, [t_1]^-) + \sum_{k=1}^{m-1} \mu([t_k]^-,[t_{k+1}]^-)
+ \mu([t_m]^-,t) \notag \\ 
&< ([T_0]^-+T_s - T_0) + (m-1)([t_0]^-+T_s - t_0)+ (T_0 - t_0) \notag \\
&< \frac{t-t_0}{n}\frac{[t_0]^-+T_s - t_0}{T_s}
< \frac{t-[t_0]^-}{n}.
\label{eq:mu(T_0,t)_case2_bound}
\end{align}
Moreover, \eqref{eq:t0-T0} gives
\begin{align*}
T_s + \frac{t-T_0}{n}- \frac{t-[t_0]^-}{n}
= T_s - \frac{T_0 - [t_0]^-}{n} 
> T_s - \frac{T_s}{n} \geq 0.
\end{align*}
Hence 
\eqref{eq:mu_T0T} follows from \eqref{eq:mu(T_0,t)_case2_bound}.
\hfill $\blacksquare$

\subsection{Proof of Proposition \ref{thm:mu_lower_bound}}
Suppose that $m \in \mathbb{N}$ satisfies $mnT_s \geq T$.

To prove the first assertion of the theorem, let
$\sigma$ have discontinuities at $knT_s + \varepsilon/m$ $(k=1,\dots,m)$.
If we define
$t = mnT_s+T_s$, then $t \geq T$ and we obtain
\begin{equation*}
\mu(0,t) = m\left(T_s - \frac{\varepsilon}{m}\right) = mT_s -\varepsilon =
\frac{1}{n}t - \left( \frac{T_s}{n} + \varepsilon \right).
\end{equation*}

To prove the second assertion, let $T_0 - [T_0]^- = \varepsilon/(2m+1)$ and let
$\sigma$ have discontinuities at
\begin{equation*}
T_0+ knT_s + \frac{\varepsilon}{2(m+1)}
= [T_0]^- + knT_s + \frac{\varepsilon}{m+1}.
\end{equation*}
for $k=1,\dots,m$.
If we let $t = T_0+mnT_s+T_s$, then $t \geq T_0 +T$ and we have
\begin{align*}
\mu(T_0,t) 
&= \left(T_s - \frac{\varepsilon}{2(m+1)}\right) + 
m\left(T_s - \frac{\varepsilon}{m+1} \right) \\
&\geq (m+1)T_s - \varepsilon \\
&= T_s + \frac{1}{n}(t-T_0) 
- \left( \frac{T_s}{n} + \varepsilon\right),
\end{align*}
which is the desired inequality \eqref{eq:lowerbound2}.
\hfill $\blacksquare$

\bibliographystyle{plain}

\end{document}